\documentclass{math-note}
\usepackage[english]{babel}
\usepackage{mathtools,amssymb,amsthm}
\usepackage{needspace}
\usepackage{microtype}
\usepackage[hidelinks,
  pdftitle={Vector Balancing in Polynomial Time},
  pdfauthor={Shengtao Guo, Ethan X. Fang, Junwei Lu},
  pdfsubject={A direct real-arithmetic algorithm for Komlos signing}
]{hyperref}
\title{\huge Vector Balancing in Polynomial Time} 
\author{}
\date{}
\numberwithin{equation}{section}
\newtheorem{theorem}{Theorem}[section]
\newtheorem{corollary}[theorem]{Corollary}
\newtheorem{lemma}[theorem]{Lemma}
\newtheorem{proposition}[theorem]{Proposition}
\theoremstyle{definition}
\newtheorem{algorithm}[theorem]{Algorithm}
\theoremstyle{remark}

\newcommand{\R}{\mathbb R}
\newcommand{\ones}{\mathbf1}
\newcommand{\Diag}{\operatorname{diag}}
\newcommand{\tr}{\operatorname{tr}}
\newcommand{\op}{\mathrm{op}}
\newcommand{\disc}{\operatorname{disc}}

\begin{document}
\author{
    Shengtao Guo\qquad
	Ethan X. Fang\thanks{Department of Biostatistics \& Bioinformatics, Duke
	University. \texttt{ethan.fang@duke.edu}.}
	\qquad
    Junwei Lu\thanks{Department of Biostatistics, Harvard T.H. Chan School of
	Public Health. \texttt{junweilu@hsph.harvard.edu}.}
}
\date{}

\maketitle
\thispagestyle{empty}

\begin{abstract}
The Koml\'os problem asks for signs
$\varepsilon\in\{-1,1\}^n$ with
$\|A\varepsilon\|_\infty$ bounded by an
absolute constant whenever the columns of
$A\in\mathbb R^{m\times n}$ have Euclidean
norm at most $1$.
We give the first polynomial-time algorithm
for this problem in the unit-cost real-RAM model.
Our deterministic algorithm achieves
discrepancy at most $8272$.
As a consequence, we obtain a deterministic
$O(\sqrt t)$-discrepancy algorithm for the
Beck--Fiala problem on set systems in which
each element belongs to at most $t$ sets.

Our algorithm rounds a fractional coloring
while keeping the row sums bounded.
Its key idea is to keep the number of
nearly tight row constraints small relative
to the number of active coordinates.
We achieve this through a spectral bound
based on a weighted Gram matrix of squared
row entries.
We maintain the bound using a cubic spectral
potential and negative-curvature updates
that satisfy the required linear constraints.
\end{abstract}

\clearpage
\setcounter{page}{1}

\section{Introduction}

Given vectors $a_1,\ldots,a_n\in\R^m$, vector balancing asks for signs
$\varepsilon_j\in\{-1,1\}$ that make their signed sum small in a
specified norm. Writing $A=[a_1\ \cdots\ a_n]$, we define its
$\ell_\infty$ discrepancy by
\[
 \disc(A)=\min_{\varepsilon\in\{-1,1\}^n}\|A\varepsilon\|_\infty.
\]
The Koml\'os conjecture asserts that $\disc(A)$ is bounded by an absolute constant whenever $\|a_j\|_2\le1$ for every $j$; see~\cite{Spe85,bansal2026discrepancy}. The Koml\'os bound would also imply an
$O(\sqrt t)$ discrepancy bound for the
Beck--Fiala problem~\cite{BF81}, where
each element belongs to at most $t$ sets.
Indeed, if $B$ is the incidence matrix
of such a set system, then $A=B/\sqrt t$
has columns of Euclidean norm at most $1$,
so $\disc(B)=O(\sqrt t)$.
Beck and Fiala proved the classical
bound of $2t-1$.

Classical partial-coloring methods
repeatedly assign signs to a constant
fraction of the remaining variables,
yielding an $O(\log n)$ discrepancy
bound for the Koml\'os problem.
Banaszczyk~\cite{Ban98} improved this to
$O(\sqrt{\log n})$ through a nonconstructive
argument based on Gaussian measures.
This bound was later matched algorithmically
by Bansal, Dadush, and Garg~\cite{BDG19},
who gave a randomized polynomial-time
algorithm.
Bansal and Jiang~\cite{BJOrig26,BJExp26}
further improved the algorithmic bound to
$(\log n)^{1/4}$ up to log-log factors,
using affine spectral independence.
Ercan~\cite{Erc26b} removed the log-log
factor in an existential bound.
The conjectured constant bound was proved
nonconstructively by Guo, Fang, and
Lu~\cite{GFL26}. Related ideas involving
densities and Fisher information appear in
Smirnov and Vershynin~\cite{SV26}.
Karingula and Lovett~\cite{KL26} gave
an elementary proof, and Bandeira~\cite{Bandeira26}
gave a short proof with constant $9\pi$.

These existence results left open the problem
of finding such a signing in polynomial time.
We give the first polynomial-time algorithm
for constant Koml\'os discrepancy in the
unit-cost real-RAM model.
The algorithm is deterministic.
Our spectral approach incorporates the energy
corrections of Bansal and Jiang~\cite{BJOrig26}
and the row-barrier techniques of Guillen
and Kobzar~\cite{GK26}.

\begin{theorem}\label{thm:main}
There is a deterministic algorithm that,
given any $m,n\ge1$ and $A\in\R^{m\times n}$
with $\|a_j\|_2\le1$ for every column
$a_j$, returns $\varepsilon\in\{-1,1\}^n$
satisfying $\|A\varepsilon\|_\infty\le8272$.
The algorithm has polynomial running time
in the unit-cost real-RAM model.
\end{theorem}

The arithmetic complexity bound is given
in Proposition~\ref{prop:arithmetic}.
As an immediate consequence of
Theorem~\ref{thm:main}, we obtain the following
algorithmic Beck--Fiala bound.

\begin{corollary}\label{cor:beck-fiala}
There is a deterministic algorithm that,
given any $B\in\{0,1\}^{m\times n}$
with at most $t$ ones in each column,
where $t\ge1$, returns
$\varepsilon\in\{-1,1\}^n$ satisfying
$\|B\varepsilon\|_\infty<16544\sqrt t$.
The algorithm has polynomial running time
in the unit-cost real-RAM model.
\end{corollary}

\subsection{Overview}

Starting from $x=0$, our algorithm
updates a fractional coloring
$x\in[-1,1]^n$, permanently freezing
each coordinate that reaches $\pm1$.
Once only a constant number of coordinates
remain active, the algorithm rounds them
to signs.
With $s$ active coordinates, only $s$
coordinates can change in an update,
while the number of rows $m$ can be
arbitrarily large. The challenge is
to make progress toward a full signing
while keeping $|(Ax)_i|$ bounded by an
absolute constant for every row $i$.

\vspace{-10pt}

\paragraph{Removing dominant entries.}
Our spectral argument requires the squared
entries retained in each row to be spread
across its coordinates. In particular,
no single retained entry should account
for too much of their total squared mass.
We enforce this condition by removing
entries that violate it and recording
their contributions at the time of removal.
Although the corresponding coordinates
may continue to change, we show that
the total error caused by these removals
remains $O(1)$.
We call the entries still tracked
in each row its \emph{retained entries}.

\vspace{-10pt}

\paragraph{Handling large and medium rows.} 
Following~\cite{BJOrig26,GK26}, we
classify rows with at least one retained
entry as \emph{large} or \emph{medium}
according to the sum of squares of
their retained entries.
Without a constraint, the contribution
of a large row's retained entries could
change substantially as the coloring evolves.
For each large row, we choose the update
direction orthogonal to the vector formed
by its retained entries, so their
contribution to the row sum remains
unchanged during the update~\cite{BF81}.
Since each active column has squared norm
at most $1$, the sum of squares of all
retained entries is at most $s$.
By choosing a sufficiently large constant
threshold, we ensure that the number of
large rows is only a small fraction of $s$.

In contrast, there may be far more medium
rows than active coordinates, so we cannot
impose an orthogonality constraint for
every one of them.
A natural idea is to constrain a medium
row only when it approaches the allowed
discrepancy bound. However, keeping its
retained contribution unchanged could
cause such constraints to accumulate,
eventually exhausting the available
degrees of freedom.
Instead, we introduce \emph{safety barriers}
that combine the row sums we track with
a concave term, adapting ideas from
\cite{BJOrig26,GK26}.
When a barrier approaches its limit,
we choose the update direction orthogonal
to its gradient. The concave term then makes the barrier
decrease whenever the update changes
a coordinate whose entry is retained
in that row.

\vspace{-10pt}

\paragraph{Limiting row constraints
with a Gram matrix.}
To leave enough freedom for each update,
we need the number of medium rows whose
safety barriers are close to their limits
to be at most a small fraction of $s$.
Following~\cite{GK26}, we assign each
safety barrier a weight that grows
exponentially as the barrier approaches
its limit.
We scale the weights so that each barrier
near its limit contributes substantially
to the total.
Controlling this total lets us bound
the number of medium-row constraints.

However, keeping the total weight from
increasing is not enough. 
As coordinates freeze and $s$ decreases,
we need the bound on the total weight
to decrease with $s$.
To achieve this, we form a weighted Gram
matrix from the squared retained entries
of medium rows.
We control the largest eigenvalue of a matrix
built from this Gram matrix.
The Gram matrix entries sum to the total
scaled weight, which is at most $s$ times
this eigenvalue.
We show that freezing a coordinate,
together with the required updates to the
row data, does not increase this eigenvalue.
Thus, bounding the largest eigenvalue
keeps the total weight at $O(s)$.
With our choice of weights, the number
of medium-row constraints remains
a small fraction of $s$.

\vspace{-10pt}

\paragraph{Choosing the update direction
via a spectral potential.}
We now need to update $x$ while keeping
the largest eigenvalue $L$ bounded.
To achieve this, we use the cubic spectral
potential $\Phi=(L-1)_+^3-\eta\|x\|_2^2$,
where $\eta>0$ is small.
The cubic term penalizes $L>1$, while
the negative quadratic rewards progress
toward a full signing.

A first-order descent direction may not exist.
For example, at the initial coloring $x=0$,
the gradient of $\Phi$ is zero.
We therefore seek a direction with almost
zero first-order change and strictly
negative curvature.
The constraints from large rows and
medium rows near their barrier limits
leave most of the $s$ degrees of freedom.
When $L\le1$, the negative quadratic
term already gives the required curvature.
When $L>1$, we show that the remaining
directions include one along which the
cubic penalty has nonpositive
second-order change.
The negative quadratic term then makes
the curvature of $\Phi$ strictly negative.

We take steps of a fixed small length,
stopping earlier if a coordinate reaches
$\pm1$. Each fixed-length step decreases
$\Phi$ by an inverse-polynomial amount.
An early stop may increase $\Phi$
slightly, but it freezes a coordinate
and can occur at most $n$ times.
Together, these facts keep $L$ bounded
and ensure polynomial termination.

\subsection{Related Work}

\paragraph{Related bounds.}
Nikolov~\cite{Nik13} proved a Koml\'os vector-coloring bound of $1$.
For signings, Kunisky~\cite{Kun23} showed that the universal constant is at least $1+\sqrt2$, and Ercan~\cite{Erc26a} obtained an existential bound below $6.84$ with convex-order guarantees.
For Beck--Fiala, improvements to general bounds~\cite{Sri97,Buk16} are complemented by results for sufficiently large sparsity~\cite{BJ25,BJOrig26} and under set-size restrictions~\cite{Alt26}.
Earlier work on online vector balancing includes
self-balancing walks~\cite{ALS21} and optimal online
discrepancy minimization~\cite{KRR24}.
At nearly logarithmic sparsity, online algorithms achieve $O(\sqrt t)$ prefix discrepancy with high probability for inputs fixed in advance~\cite{AT26,Hmadi26}; Hmadi assumes exact sampling and potential evaluation.
Complementary results examine input randomness through fixed-weight column models~\cite{BM20} and Gaussian-smoothed Koml\'os instances~\cite{BJMSS22}.

\paragraph{Coloring algorithms.}
Partial-coloring arguments build a full signing in rounds~\cite{Bec81,Spe85,gluskin1989extremal}.
Algorithms use semidefinite programming~\cite{Ban10,BS13}, random walks~\cite{LM15}, and convex geometry~\cite{Rot17,ES18}.
Controlling error across rounds motivates the framework of Bansal and Garg~\cite{BG17}; Harvey, Schwartz, and Singh~\cite{HSS14} compute a full coloring directly for Spencer's theorem.
More broadly, Bansal~\cite{Ban24} combines iterated rounding with concentration when a constant fraction of active dimensions remains free.
A second route uses sub-Gaussian signing distributions associated with Banaszczyk's theorem~\cite{DGLN19}: the Gram--Schmidt walk constructs such distributions~\cite{GSW19}, and subsequent analysis controls how bounded numerical errors affect concentration~\cite{AvC26}.
Further work improves running times with hereditary guarantees~\cite{Lar23} and achieves nearly input-sparsity time for Spencer's theorem~\cite{JSS23}.

\paragraph{Spectral methods.}
Potential methods use exponential weights~\cite{LRR17} or row slacks and remaining variance~\cite{BLV22}.
In spectral sparsification, eigenvalue barriers~\cite{BSS12} have a regularization and regret-minimization interpretation~\cite{ALO15}.
Pesenti and Vladu~\cite{PV23} use this regularization viewpoint to guide deterministic discrepancy walks.
Lau, Wang, and Zhou~\cite{LWZ25} extended their walk to matrix discrepancy and spectral sparsification.
Spectral assumptions on the input also yield bounds for regular hypergraphs~\cite{Pot20} and, through the entrywise-square matrix, pseudorandom Koml\'os instances~\cite{PV23}.
Our algorithm bounds the largest eigenvalue of a matrix built from the weighted Gram matrix of squared retained entries.
This keeps the number of nearly tight row constraints below a fixed fraction of the active coordinates throughout rounding.

\paragraph{Subsequent developments.}
Subsequent work gave a Brownian-walk exposition of the spectral potential in the Beck--Fiala setting~\cite{BansalGFL26} and faster deterministic algorithms in real arithmetic with smaller discrepancy constants~\cite{Li26,AkbasSra26}.
For rational inputs, Li~\cite{Li26} established polynomial bit complexity, and Anand and Chen~\cite{AnandChen26} obtained a randomized algorithm with improved expected bit complexity.

\section{The Algorithm}
\label{sec:algorithm}

\subsection{Fractional coloring and entry removal}

Starting from $x=0$, our algorithm
maintains a fractional coloring
$x\in[-1,1]^n$. Once a coordinate
reaches $\pm1$, its value is fixed
for the rest of the algorithm.
We call the coordinates with $|x_j|<1$
\emph{active}, and write
$V=\{j\in[n]:|x_j|<1\}$ and $s=|V|$.
Each update changes only the active
coordinates.

For each row $i$, we maintain a set
$T_i\subseteq V$ of column indices
whose entries are retained.
Initially, $T_i=\{j:a_{ij}\ne0\}$.
We also maintain a number $c_i$,
initially zero, to record the
contributions of removed entries.
Whenever an entry $a_{ij}$ is removed,
we perform
\begin{equation}\label{eq:cleanup}
c_i\leftarrow c_i+a_{ij}x_j,\qquad
T_i\leftarrow T_i\setminus\{j\}.
\end{equation}
The row sum used by the algorithm is
$d_i=c_i+\sum_{j\in T_i}a_{ij}x_j$.
Thus, removing an entry leaves $d_i$
unchanged at the current coloring.

We remove entries in two cases.
First, when a coordinate reaches $\pm1$,
we remove its entry from every row
where it is retained.
Second, for each row $i$, let
$q_i=\sum_{j\in T_i}a_{ij}^2$.
Fix $K=16$. We repeatedly remove
any retained entry satisfying
$a_{ij}^2>q_i/K$, recomputing $q_i$
after each removal.
The resulting entries satisfy
\begin{equation}\label{eq:spread}
a_{ij}^2\le q_i/K
\qquad\text{for all }j\in T_i.
\end{equation}
We call these removals
\emph{dominant-entry removals}.

An entry removed by the second rule
may still have a changing coordinate.
Thus, $d_i$ may differ from the true
row sum $(Ax)_i$.
We will prove that the sum of the
absolute values of deliberately
removed entries in each row is at
most $2K$. Since every coordinate
stays in $[-1,1]$, this gives
$|(Ax)_i-d_i|\le4K$ before the
final rounding.

\subsection{Row constraints and safety barriers}

We use $q_i$ to decide how to handle
each row. A row is \emph{large} if
$q_i>K^2$, \emph{medium} if
$0<q_i\le K^2$, and \emph{completed}
if $q_i=0$. A completed row has no
retained entries and requires no
further constraints.

Let $h\in\R^V$ denote an update direction.
For each large row $i$, we require
\[
\sum_{j\in T_i}a_{ij}h_j=0.
\]
This keeps $d_i$ unchanged during the
update. Every large row starts with
$d_i=0$, and the orthogonality
constraint preserves this value.
Since $q_i$ never increases, a row
cannot become large after leaving
the large class.
Moreover, $\sum_iq_i\le s$, so there
are fewer than $s/K^2$ large rows.
Thus, these constraints leave most
of the $s$ degrees of freedom available.

Following~\cite{BJOrig26,GK26}, we use
a concave term to help decrease the
safety barriers of medium rows.
For each medium row $i$, define
$G_i=\sum_{j\in T_i}a_{ij}^2(1-x_j^2)$.
Since $|x_j|\le1$, we have
$0\le G_i\le q_i$.
Moreover, $G_i$ is concave in the
active coordinates.

For each medium row $i$, we maintain
a scale $r_i$ satisfying
\begin{equation}\label{eq:scale}
q_i\le r_i^2<4q_i,\qquad r_i\le K.
\end{equation}
Once initialized, this scale never
increases.
To bound $|d_i|$, we define, for
$\sigma\in\{-1,1\}$,
\begin{equation}\label{eq:barrier}
u_{i\sigma}
=\frac{\sigma d_i-2K^3}{r_i}
+\frac{G_i}{2r_i^2}+12K.
\end{equation}
We call $u_{i\sigma}$ a
\emph{safety barrier} and maintain
$u_{i\sigma}\le0$ for both signs
throughout the algorithm.
The precise rule for updating $r_i$
is given in the cleanup procedure below.

We call a safety barrier \emph{nearly tight}
if $u_{i\sigma}\ge-1$.
For each nearly tight barrier, we require
the update direction $h$ to satisfy
\[
\langle\nabla u_{i\sigma}(x),h\rangle=0,
\]
where the gradient is taken with respect
to the active coordinates.
This removes the first-order change.
The concavity of the barrier then
ensures that it cannot increase during
the update.

\subsection{Gram matrix and spectral potential}

To limit the number of nearly tight
barriers, we assign weights to the
barriers and collect them in a matrix.
For each row $i$, let $b_i\in\R^V$
denote its retained row vector, with
$(b_i)_j=a_{ij}$ for $j\in T_i$
and zero otherwise.
Define $p_i\in\R^V$ by
$(p_i)_j=(b_i)_j^2$.
Thus $\sum_{j\in V}(p_i)_j=q_i$.

Following~\cite{BJOrig26,GK26}, we use
exponential weights for the safety barriers.
Our key new idea is to control these
weights through the largest eigenvalue
of a weighted Gram matrix.
For each medium row $i$ and sign
$\sigma\in\{-1,1\}$, define
\begin{equation}\label{eq:matrix-potential}
w_{i\sigma}=\frac{K^2}{2r_i^4}
e^{u_{i\sigma}/(2K)},\qquad
W=\sum_{\substack{i\text{ medium}\\
\sigma=\pm1}}
w_{i\sigma}p_ip_i^{\mathsf T}.
\end{equation}
The exponential factor gives more weight
to barriers near zero, while $r_i^{-4}$
accounts for differences in row size.
Thus, $W$ collects the weighted
contributions from all medium rows.

As entries are removed, a large row
may become medium, adding new terms
to $W$. To handle this change, we add
a diagonal term for large rows that
is large enough to cover these new
terms. We also add a small rank-one
term to make the largest eigenvalue
simple, so its derivatives are well
defined.
Set $\eta=1/(Kn)$ and define
\begin{equation}\label{eq:full-matrix}
\mathcal B=\frac1K\sum_{i\text{ large}}
\Diag(p_i),\qquad
M=W+\mathcal B+\eta\ones\ones^{\mathsf T},
\end{equation}
where $\Diag(p_i)$ denotes the diagonal
matrix with diagonal $p_i$, and $\ones$
is the all-ones vector in $\R^V$.
We write $L=\lambda_{\max}(M)$ when
$V\ne\varnothing$, and set $L=0$
when $V=\varnothing$.

To keep $L$ bounded while moving $x$
toward a full signing, we use the potential
\begin{equation}\label{eq:potential}
\Phi=(L-1)_+^3-\eta\|x\|_2^2,
\end{equation}
where $t_+=\max\{t,0\}$.
The cubic term penalizes $L>1$, while
the negative quadratic rewards progress
toward signs. The cubic penalty has a
continuous second derivative at $L=1$,
which allows us to use second derivatives
when choosing update directions.

\subsection{Cleanup and update rule}

Each iteration ends with a \emph{cleanup}
step, which prepares the row data for the
next update. We also perform cleanup
at initialization.

We first transfer the retained
contributions of frozen coordinates
to $c_i$, remove their entries from
the sets $T_i$, and then remove
these coordinates from $V$.
We recompute $q_i$ and apply the
dominant-entry removal rule.
When several entries qualify for removal,
we choose the lexicographically smallest
pair $(i,j)$ and recompute the affected
$q_i$ after each removal.
Once no entry qualifies, we reclassify
the rows and update the scales $r_i$.
Each removed entry's current contribution
is added to $c_i$, so cleanup leaves
$d_i$ unchanged.

To update the scales of medium rows,
let $a_i^*=\max_{j\in T_i}|a_{ij}|$.
Set $r=a_i^*$ and repeatedly double
$r$ until $r^2\ge q_i$.
If row $i$ has just become medium,
set $r_i=\min\{K,r\}$.
Otherwise, replace $r_i$ by
$\min\{r_i,r\}$.
At initialization, all medium rows
are treated as newly medium.
These updates ensure
$q_i\le r_i^2<4q_i$ and $r_i\le K$.
Once a scale is initialized, it
never increases.

With the row data updated, we choose
a direction $h\in\R^V$ satisfying
the large-row and nearly tight barrier
constraints. We also seek to make the
first-order change of $\Phi$ small and
its second-order change negative.
When taking derivatives, we keep
$V$, $T_i$, $c_i$, and $r_i$ fixed.

Set $H=2^{24}$ and choose
$\tau=\eta/(KHn^2)$ and
$\epsilon=\eta\tau/K^2$.
The value of $H$ follows from
the Taylor estimate in
Appendix~\ref{sec:taylor}.
The following lemma gives the precise
guarantees for computing an update direction.
Its proof is given in
Appendix~\ref{sec:direction-implementation}.

\begin{lemma}\label{lem:direction-finding}
At any state after cleanup with $s>K$,
$L\le2$, and all medium-row barriers
nonpositive, a deterministic procedure
finds $h\in\R^V$ satisfying the
large-row and nearly tight barrier
constraints, together with
\[
\frac12\le\|h\|_2\le1,\qquad
|\langle\nabla\Phi,h\rangle|\le\epsilon,
\qquad
h^{\mathsf T}\nabla^2\Phi h
\le-\frac{3\eta}{16}.
\]
The procedure uses polynomially many
real arithmetic operations and comparisons.
\end{lemma}

Given such a direction, we take a step
of length at most $\tau$.
If the step ends early, we choose between
$h$ and $-h$ so that a coordinate reaches
$\pm1$ and becomes fixed.
Together with cleanup, this gives
Algorithm~\ref{alg:signing}.

\begin{algorithm}\label{alg:signing}
If $n\le K$, return all positive signs.
Otherwise, initialize $x=0$, $V=[n]$,
$c_i=0$, and
$T_i=\{j\in[n]:a_{ij}\ne0\}$
for each row $i$.
Perform cleanup, treating all medium
rows as newly medium.

While $|V|>K$, repeat:
\begin{enumerate}

\item Compute an update direction $h$
using Lemma~\ref{lem:direction-finding}.

\item Set $t=\min\{\tau,t_*\}$, where
\[
t_*=\min_{j\in V:h_j\ne0}
\frac{1-|x_j|}{|h_j|}.
\]
If $t=\tau$, set $\xi=1$.
Otherwise, let $j_*$ be the smallest
index attaining $t_*$, and set
$\xi=1$ if $x_{j_*}h_{j_*}\ge0$,
and $\xi=-1$ otherwise.

\item Update $x_V\leftarrow x_V+\xi th$
and perform cleanup.

\end{enumerate}

Round each remaining active coordinate
to the nearest sign, choosing $+1$
in case of a tie.
Return the resulting signing.
\end{algorithm}

\subsection{Proof Structure}
\label{sec:proof-structure}

Section~\ref{sec:removal} bounds the tracking error from entry
removal and shows that cleanup preserves the row bounds
without increasing the potential.
Section~\ref{sec:curvature} then proves the existence of a
negative-curvature direction satisfying the row constraints.
Appendix~\ref{sec:computation} gives a deterministic procedure
to compute such a direction in polynomially many real
arithmetic operations and comparisons.

Using this procedure and the Taylor estimate from
Appendix~\ref{sec:taylor}, Section~\ref{sec:progress} shows
that the required bounds remain valid and that the algorithm
terminates after $O(n^7)$ updates.
The tracked row bounds, together with the entry-removal and
final rounding errors, give the discrepancy bound in
Theorem~\ref{thm:main}.
Appendix~\ref{sec:complexity} combines the update count with
the cost per update to prove the following arithmetic bound.

\begin{proposition}\label{prop:arithmetic}
Algorithm~\ref{alg:signing} can be
implemented deterministically using
\[
O\bigl((mn^9+n^{10})\log(2+m+n)\bigr)
\]
additions, subtractions, multiplications, divisions,
and comparisons on real numbers.
\end{proposition}

\section{Entry Removal and Cleanup}
\label{sec:removal}

Two safeguards keep the algorithm ready for each update:
removing dominant entries and cleaning up the row data.
Dominant-entry removal enforces the spreading condition
\eqref{eq:spread} needed for the spectral analysis.
A removed entry's coordinate may still change, so our first
lemma bounds the resulting error in the tracked row sum $d_i$.
We then show that cleanup preserves the bounds needed
for the next update.

\Needspace{7\baselineskip}
\begin{lemma}\label{lem:deletion}
The sum of the absolute values of entries
removed by the dominant-entry rule
from any row is at most $2K$.
At every state before the final rounding,
we have $|(Ax)_i-d_i|\le4K$ for every row $i$.
\end{lemma}

\begin{proof}
Fix a row $i$. Only dominant-entry removals can create tracking
error, so we first bound the total magnitude removed by this rule.
If there are no such removals, this total is zero. Otherwise, at the
first removal let $q$ be the retained squared mass and $a$ the entry
being removed. Since $a^2>q/K$ and $|a|\le1$, we have $q<K$.
The retained mass can only decrease thereafter.

For each dominant-entry removal, let $Q$ be the squared mass just
before the removal and $a$ the removed entry. Dominance gives
$\sqrt Q<\sqrt K\,|a|$, and hence
$\sqrt Q-\sqrt{Q-a^2}
=a^2/(\sqrt Q+\sqrt{Q-a^2})>|a|/(2\sqrt K)$.
Let $a_\ell$ and $Q_\ell$ denote these values for the $\ell$-th
removal. Other removals only decrease the mass, so
$Q_{\ell+1}\le Q_\ell-a_\ell^2$. The decreases telescope:
\[
\begin{aligned}
\sum_\ell |a_\ell|
&\le 2\sqrt K\sum_\ell
\bigl(\sqrt{Q_\ell}-\sqrt{Q_\ell-a_\ell^2}\bigr)\\
&\le2\sqrt K\sqrt{Q_1}<2K.
\end{aligned}
\]

We now bound the tracking error. Let $D_i$ contain the entries
removed by the dominant-entry rule, and let $x_j^{(i,j)}$ be the
value of $x_j$ when $a_{ij}$ was removed. The offset updates give
$(Ax)_i-d_i=\sum_{j\in D_i}a_{ij}(x_j-x_j^{(i,j)})$.
Entries removed when a coordinate freezes cause no error.
Since $x_j,x_j^{(i,j)}\in[-1,1]$,
\[
|(Ax)_i-d_i|\le2\sum_{j\in D_i}|a_{ij}|\le4K.
\qedhere
\]
\end{proof}

Cleanup leaves $x$ fixed but changes the retained entries,
row scales, and possibly the row classification.
A large row may become medium, adding new terms to $W$.
The next lemma shows that these changes preserve the
safety barriers and cannot increase $\Phi$.

\begin{lemma}\label{lem:cleanup}
After the initial cleanup,
$L\le 2/K<1$ and $\Phi=0$.

Suppose that, before cleanup,
$d_i=0$ for every large row $i$, and
$u_{i\sigma}\le0$ for every medium row $i$
and both signs $\sigma\in\{-1,1\}$.
Then cleanup satisfies:
\begin{enumerate}
\item Every tracked row sum $d_i$
remains unchanged.
\item The safety barriers of rows
that remain medium do not increase.
\item Every newly medium barrier
is strictly below $-1$.
\item The spectral potential $\Phi$
does not increase.
\end{enumerate}
\end{lemma}

To show that cleanup cannot increase $L$, the key case is
when a large row becomes medium. Its contribution to
$\mathcal B$ disappears, while two weighted terms enter $W$.
We must bound these new terms by the diagonal contribution they
replace. The new medium row inherits $d_i=0$ from its large
phase. Consequently,
$u_{i\sigma}\le-2K^2+12K+1/2<-1$, and its weights satisfy
$2r_i^2w_{i\sigma}<1/K$; see Appendix~\ref{app:cleanup}.
Weighted Cauchy--Schwarz gives
$p_ip_i^{\mathsf T}\preceq q_i\Diag(p_i)$, and
$q_i\le r_i^2$. Therefore,
\[
\sum_{\sigma=\pm1}w_{i\sigma}p_ip_i^{\mathsf T}
\preceq\frac1K\Diag(p_i).
\]
The diagonal term assigned while the row was large
therefore covers its new contribution to $W$.

For a row that stays medium, removing entries decreases
$G_i$ at fixed scale, and decreasing its scale also
reduces its barriers and weights. Its contribution to
$W$ thus decreases entrywise. Since $M$ has nonnegative
entries, these changes cannot increase its largest
eigenvalue. Freezing a coordinate restricts $M$ to a
principal submatrix, and the remaining large-row diagonal
terms also decrease. Hence cleanup cannot increase $L$
or $\Phi$, since $x$ stays fixed. At initialization the
same estimates give $W+\mathcal B\preceq I/K$, so
$L\le1/K+\eta n=2/K$ and $\Phi=0$, as $x=0$.
The scale and weight estimates are proved in
Appendix~\ref{app:cleanup}.

\section{Finding a Negative-Curvature Direction}
\label{sec:curvature}

The previous section showed that cleanup preserves the required bounds
and cannot increase $\Phi$. We now consider how to update $x$ between
cleanup steps. Our goal is to find a direction satisfying the row
constraints along which $\Phi$ has strictly negative curvature.

We first bound the number of constraints from large rows and nearly
tight barriers. We then use the Gram matrix and the concavity of the
safety barriers to establish the existence of such a direction.
This gives the main mathematical guarantee needed for
Lemma~\ref{lem:direction-finding}.

\subsection{Counting row constraints}

Each large row and each nearly tight barrier imposes one linear
constraint. Their number is limited by two bounded quantities:
$\sum_iq_i\le s$ for large rows, and
$\ones^{\mathsf T}W\ones\le Ls$ for nearly tight barriers.
Indeed, every nearly tight pair contributes a substantial amount
to the latter quantity.

\begin{lemma}\label{lem:row-count}
At any state after cleanup with $s>K$ and $L\le2$, the number of
large rows plus the number of nearly tight pairs $(i,\sigma)$ is
less than $5s/K$.
\end{lemma}

\begin{proof}
Since $\sum_iq_i\le s$ and $q_i>K^2$ for every large row,
there are fewer than $s/K^2$ large rows.

We next count the nearly tight pairs. For each medium row $i$
and sign $\sigma\in\{-1,1\}$, the definition of $w_{i\sigma}$ gives
\[
w_{i\sigma}q_i^2
=\frac{K^2}{2}e^{u_{i\sigma}/(2K)}
 \left(\frac{q_i}{r_i^2}\right)^2.
\]
By \eqref{eq:scale}, $1/4<q_i/r_i^2\le1$.
Thus every nearly tight pair, for which $u_{i\sigma}\ge-1$,
contributes more than $K^2e^{-1/(2K)}/32$.

On the other hand, since $W\preceq M$ and $L\le2$,
\[
\sum_{i\text{ medium},\,\sigma=\pm1}w_{i\sigma}q_i^2
=\ones^{\mathsf T}W\ones\le Ls\le2s.
\]
Hence the number of nearly tight pairs is less than
$64e^{1/(2K)}s/K^2$. Including the large rows, the total
number of row constraints is less than
$(1+64e^{1/(2K)})s/K^2<5s/K$,
where the last inequality holds for $K=16$.
\end{proof}

\subsection{Controlling large eigenvalues}

To find a direction of negative curvature for $\Phi$, we need to
control the second derivative of $L=\lambda_{\max}(M)$.
As $x$ changes, the top eigenvector of $M$ may change as well,
adding a nonnegative term to $L''$. This term can be large when
other eigenvalues are close to $L$. We will eliminate these
contributions using linear constraints, so we need to know how
many such eigenvectors there are.

The bound comes from the trace. The spreading condition
\eqref{eq:spread} makes $\tr W$ at most
$\ones^{\mathsf T}W\ones/K$. Since the trace of $M$ is the
sum of its eigenvalues, only a few eigenvalues can be at least $L/2$.

\begin{lemma}\label{lem:eigenvalue-count}
At any state after cleanup with $s>K$ and $1<L\le2$,
let $F$ be the subspace spanned by the eigenvectors of $M$
with eigenvalues at least $L/2$. Then $\dim F\le4s/K+1/2$.
\end{lemma}

\begin{proof}
The entry-removal rule \eqref{eq:spread} implies
$\|p_i\|_2^2\le(\max_j(p_i)_j)q_i\le q_i^2/K$.
Consequently,
\begin{equation}\label{eq:gram-trace}
\tr W
=\sum_{i\text{ medium},\,\sigma=\pm1}
 w_{i\sigma}\|p_i\|_2^2
\le\frac1K\ones^{\mathsf T}W\ones
\le\frac{Ls}{K}.
\end{equation}
For the remaining terms,
$\tr\mathcal B\le s/K$ and
$\tr(\eta\ones\ones^{\mathsf T})=\eta s$.
Thus $\tr M\le Ls/K+s/K+\eta s$.

Since $M$ is positive semidefinite, each eigenvalue corresponding
to $F$ contributes at least $L/2$ to the trace. Therefore,
\[
\dim F
\le\frac{2\tr M}{L}
\le\frac{4s}{K}+2\eta s
\le\frac{4s}{K}+\frac12.
\]
Here we used $L>1$, $s\le n$, and $\eta=1/(Kn)$ with $K=16$.
\end{proof}

The rank-one term $\eta\ones\ones^{\mathsf T}$ also separates
the largest eigenvalue from the others. Since $W+\mathcal B$ is
entrywise nonnegative, its top eigenvector can be chosen nonnegative.
The rank-one term raises the corresponding Rayleigh quotient by
at least $\eta$, leading to the following spectral gap.

\begin{lemma}\label{lem:spectral-gap}
For $s\ge2$, the largest two eigenvalues of $M$ satisfy
$\lambda_1(M)-\lambda_2(M)\ge\eta$.
For every $s>0$, a unit eigenvector corresponding to
$L=\lambda_1(M)$ can be chosen strictly positive.
\end{lemma}

\begin{proof}
Write $N=W+\mathcal B$, so that
$M=N+\eta\ones\ones^{\mathsf T}$.
Since $N$ is symmetric and entrywise nonnegative, it has a
nonnegative unit eigenvector $v_0$ corresponding to
$\lambda_1(N)$. In particular, $\ones^{\mathsf T}v_0\ge1$.
Therefore,
\[
\lambda_1(M)
\ge v_0^{\mathsf T}Mv_0
=\lambda_1(N)+\eta(\ones^{\mathsf T}v_0)^2
\ge\lambda_1(N)+\eta.
\]
On the other hand, the rank-one term vanishes on
$\ones^\perp$. By the min--max principle,
\[
\lambda_2(M)
\le\max_{\substack{\|z\|_2=1\\z\perp\ones}}
 z^{\mathsf T}Mz
\le\lambda_1(N).
\]
Combining the inequalities gives the gap. Every entry of
$M$ is strictly positive, so its top unit eigenvector can be
chosen strictly positive.
\end{proof}

\subsection{Existence of negative curvature}

We now combine the counting bounds to find a direction of negative
curvature. The row constraints eliminate few directions. When $L>1$,
the eigenvalue bound lets us control the top eigenvector using only
a few additional constraints. In the remaining subspace, we will
compare the positive curvature of the exponential weights with
the negative curvature of the concave safety barriers. The term
$-\eta\|x\|_2^2$ then makes the curvature of $\Phi$ strictly
negative.

We will use the following lemma to choose a direction after
bounding the curvature. If a positive semidefinite quadratic form
is small on each coordinate compared with a diagonal form, then
some direction in every sufficiently large subspace has
nonpositive difference. The result follows from a trace argument
\cite[Lemma~4.2]{GK26}; see also
\cite[Appendix~B.2]{LRR17}.

\begin{lemma}\label{lem:restricted-trace}
Let $S\succeq0$ and let $D\succeq0$ be diagonal on $\R^s$.
Suppose $S_{jj}\le\rho D_{jj}$ for every $j$, where $\rho\ge0$.
Every subspace $\mathcal H\subseteq\R^s$ with
$\dim\mathcal H>\rho s$ contains a nonzero vector $h$ with
$h^{\mathsf T}(S-D)h\le0$.
\end{lemma}

The proof is given in Appendix~\ref{sec:linear-algebra}.

\begin{proposition}\label{prop:direction}
At any state after cleanup with $s>K$ and $L\le2$, for every
$g_*\in\R^V$, there is a unit vector $h\in\R^V$ orthogonal
to $g_*$, all large-row vectors, and all nearly tight barrier
gradients, such that
\[
h^{\mathsf T}\nabla^2\Phi\,h\le-2\eta.
\]
\end{proposition}

\begin{proof}
We keep $V,T_i,c_i,r_i$ fixed when differentiating with respect
to $x_V$. If $L\le1$, the cubic term in \eqref{eq:potential}
has zero Hessian, so $\nabla^2\Phi=-2\eta I$.
By Lemma~\ref{lem:row-count}, the row constraints and
$g_*^{\mathsf T}h=0$ leave a subspace of dimension greater than
$s-5s/K-1>0$. Any unit vector in this subspace suffices.

Suppose $L>1$. We seek a unit direction $h$ with
$L'[h]=0$ and $L''[h,h]\le0$; then the spectral term in
$\Phi$ has nonpositive curvature and the negative quadratic
term provides the required strict bound.

Let $v$ be the positive unit top eigenvector of $M$.
Lemma~\ref{lem:spectral-gap} ensures that $L$ is simple.
During an update, $M'[h]=W'[h]$ and $M''[h,h]=W''[h,h]$,
because $\mathcal B$ and $\eta\ones\ones^{\mathsf T}$ are fixed.
For an orthonormal eigenbasis $v,v_2,\ldots,v_s$ of $M$,
the eigenvalue perturbation formula gives
\begin{equation}\label{eq:spectral-curvature}
L''[h,h]
=v^{\mathsf T}W''[h,h]v
+2\sum_{\ell=2}^s
\frac{|v_\ell^{\mathsf T}W'[h]v|^2}
{L-\lambda_\ell(M)}.
\end{equation}
The sum can be large when another eigenvalue approaches $L$.
We first eliminate those terms by imposing linear constraints.

Let $F$ be the subspace from Lemma~\ref{lem:eigenvalue-count},
and let $\mathcal H$ consist of directions satisfying all
large-row and nearly tight barrier constraints, together with
$g_*^{\mathsf T}h=0$ and
$\operatorname{Proj}_F W'[h]v=0$.
The last condition is linear in $h$ and eliminates the terms
in \eqref{eq:spectral-curvature} whose eigenvalues are at least
$L/2$. Since $v\in F$, it also gives
$L'[h]=v^{\mathsf T}W'[h]v=0$.

By Lemmas~\ref{lem:row-count} and~\ref{lem:eigenvalue-count},
\[
\dim\mathcal H>\left(1-\frac9K\right)s-\frac32>\frac{5s}{K}.
\]
The last inequality uses $K=16$ and $s>K$.
Thus enough directions remain for
Lemma~\ref{lem:restricted-trace}.

For $h\in\mathcal H$, every remaining denominator in
\eqref{eq:spectral-curvature} exceeds $L/2$, so
\begin{equation}\label{eq:spectral-curvature-reduced}
L''[h,h]
\le v^{\mathsf T}W''[h,h]v
 +\frac4L\|W'[h]v\|_2^2.
\end{equation}
We bound the two terms on the right using the concavity
of the barriers and the Gram structure of $W$.

For a medium row, differentiating \eqref{eq:barrier} yields
\begin{equation}\label{eq:row-derivatives}
g_{i\sigma}:=\nabla u_{i\sigma}
=\frac{\sigma b_i}{r_i}-\frac{p_i\circ x_V}{r_i^2},
\qquad
\nabla^2u_{i\sigma}=-\frac1{r_i^2}\Diag(p_i),
\end{equation}
where $\circ$ denotes coordinatewise multiplication.
The spreading and scale conditions also give
\begin{equation}\label{eq:gradient-bounds}
\|g_{i\sigma}\|_2\le2,\qquad
(g_{i\sigma})_j^2\le\frac{2(p_i)_j}{r_i^2},\qquad
\|\nabla^2u_{i\sigma}\|_{\op}\le\frac1K.
\end{equation}
Indeed, \eqref{eq:spread} and \eqref{eq:scale} imply
$\|b_i\|_2/r_i\le1$,
$\|p_i\|_2/r_i^2\le1/\sqrt K$, and
$\max_j(p_i)_j/r_i^2\le1/K$.
For the coordinate bound, set
$z=|(b_i)_j|/r_i\le1/\sqrt K\le1/4$.
Then $|(g_{i\sigma})_j|\le z+z^2$ and
$(z+z^2)^2\le2z^2$.

Put $\beta=1/(2K)$. Differentiating the weights twice
produces a nonnegative quadratic term in
$g_{i\sigma}^{\mathsf T}h$ and a nonpositive term from
$\nabla^2u_{i\sigma}$. To compare them, define
\[
\mathcal U=\sum_{i,\sigma}w_{i\sigma}\beta^2
 (p_i^{\mathsf T}v)^2
 g_{i\sigma}g_{i\sigma}^{\mathsf T},
\qquad
\mathcal D=\sum_{i,\sigma}\frac{w_{i\sigma}\beta}{r_i^2}
 (p_i^{\mathsf T}v)^2\Diag(p_i),
\]
where the sums range over medium rows and both signs.
Direct differentiation gives
$v^{\mathsf T}W''[h,h]v
=h^{\mathsf T}(\mathcal U-\mathcal D)h$.
We will use the bounds
\begin{equation}\label{eq:curvature-diagonal}
\mathcal U_{jj}\le\frac1K\mathcal D_{jj}
\quad\text{for every }j,
\qquad
\|W'[h]v\|_2^2\le L\,h^{\mathsf T}\mathcal U h.
\end{equation}
The first follows from \eqref{eq:gradient-bounds} and
$\beta=1/(2K)$. For the second, let $\mathcal C$ have columns
$\sqrt{w_{i\sigma}}\,p_i$, and put
$z_{i\sigma}=\sqrt{w_{i\sigma}}\,\beta
(g_{i\sigma}^{\mathsf T}h)(p_i^{\mathsf T}v)$.
Then $W=\mathcal C\mathcal C^{\mathsf T}$,
$W'[h]v=\mathcal C z$, and
$\|z\|_2^2=h^{\mathsf T}\mathcal U h$.
The bound follows since
$\|\mathcal C\|_{\op}^2=\|W\|_{\op}\le L$.

Combining \eqref{eq:spectral-curvature-reduced} with these estimates
gives $L''[h,h]\le h^{\mathsf T}(5\mathcal U-\mathcal D)h$
for $h\in\mathcal H$.
The diagonal bound \eqref{eq:curvature-diagonal} implies
$5\mathcal U_{jj}\le(5/K)\mathcal D_{jj}$.
Since $\dim\mathcal H>5s/K$,
Lemma~\ref{lem:restricted-trace} gives a nonzero direction
$h\in\mathcal H$ with $L''[h,h]\le0$.
Normalize $h$ to unit length.

Finally, $L'[h]=0$ and $L''[h,h]\le0$, so the chain rule
gives $h^{\mathsf T}\nabla^2\Phi\,h
=f'(L)L''[h,h]-2\eta\le-2\eta$ for
$f(t)=(t-1)_+^3$, as required.
\end{proof}

\section{Progress and Termination}
\label{sec:progress}

We now show that the algorithm can keep taking updates until only
$K$ coordinates remain active. To apply
Lemma~\ref{lem:direction-finding} at every iteration, we must keep
$L<2$ and all medium-row barriers nonpositive. The row constraints
protect the barriers during an update. To control $L$, we use the
potential $\Phi$: a \emph{full step} ($t=\tau$) decreases it, while
a \emph{short step} ($t<\tau$) fixes a coordinate and can increase it
only slightly. There are at most $n$ short steps.

\begin{proposition}\label{prop:progress}
After initialization and after every cleanup in
Algorithm~\ref{alg:signing}, the state satisfies
\[
d_i=0\quad\text{for large rows},\qquad
u_{i\sigma}\le0\quad\text{for medium pairs},\qquad L<2.
\]
The algorithm reaches $s\le K$ after $O(n^7)$ updates.
\end{proposition}

\begin{proof}
The bounds hold initially by Lemma~\ref{lem:cleanup}.
Assume they hold at a state after cleanup with $s>K$.
Lemma~\ref{lem:direction-finding} gives a direction $h$
satisfying the row constraints,
$\|h\|_2\le1$,
$|\langle\nabla\Phi,h\rangle|\le\epsilon$, and
$h^{\mathsf T}\nabla^2\Phi h\le-3\eta/16$.

With the row data fixed, the large-row constraint
$b_i^{\mathsf T}h=0$ keeps $d_i=0$.
For a nearly tight medium barrier,
$\langle\nabla u_{i\sigma},h\rangle=0$ and the exact
quadratic form of \eqref{eq:barrier} give
\[
u_{i\sigma}(x+\xi th)
=u_{i\sigma}(x)-\frac{t^2}{2r_i^2}
 \sum_{j\in V}(p_i)_j h_j^2
\le u_{i\sigma}(x).
\]
Other medium barriers start below $-1$ and, by
\eqref{eq:gradient-bounds}, have gradient norm at most $2$
throughout the update segment. Since $t\le\tau$ and $2\tau<1$,
they remain negative. Lemma~\ref{lem:cleanup} preserves these
row bounds during cleanup.

Set $\Delta=\eta\tau^2/K$.
Lemma~\ref{lem:smoothness} applies at the starting state,
where $L<2$, and gives for either sign $\xi\in\{-1,1\}$
and every $0\le t\le\tau$,
\begin{equation}\label{eq:progress-taylor}
\Phi(x+\xi th)-\Phi(x)
\le \epsilon t-\frac{3\eta}{32}t^2
 +\frac{Hn^2}{6}t^3.
\end{equation}
For a full step, $\epsilon=\eta\tau/K^2$ and
$Hn^2\tau=\eta/K$ bound the potential change by
$\eta\tau^2(1/K^2-3/32+1/(6K))\le-\Delta$, since $K=16$.
For a short step, the last two terms in
\eqref{eq:progress-taylor} have nonpositive sum because
$Hn^2t\le\eta/K$, so the increase is at most
$\epsilon\tau=\Delta/K$.
For such a step, the sign choice in Algorithm~\ref{alg:signing} gives
$|x_{j_*}+\xi th_{j_*}|=1$, freezing at least one coordinate.
Cleanup cannot increase $\Phi$ by Lemma~\ref{lem:cleanup}.

Every short step freezes a new coordinate, so at most $n$ occur.
Starting from $\Phi=0$ and ignoring full-step decreases gives
$\Phi\le n\Delta/K$ after cleanup.
Since $\|x\|_2^2\le n$ and $\eta n=1/K$,
\[
(L-1)_+^3=\Phi+\eta\|x\|_2^2
\le\frac1K\left(1+\frac{\tau^2}{K^2}\right)<1.
\]
Thus $L<2$ at the next state. The Taylor estimate required
$L<2$ only at the start of the step, so this closes the induction.

For any finite prefix, let $N_{\mathrm{full}}$ and
$N_{\mathrm{short}}$ count its full and short steps.
Since $\Phi\ge-\eta n=-1/K$, the potential budget gives
$-1/K\le\Phi\le-N_{\mathrm{full}}\Delta
+N_{\mathrm{short}}\Delta/K$, with $N_{\mathrm{short}}\le n$.
Hence $N_{\mathrm{full}}+N_{\mathrm{short}}
\le1/(K\Delta)+(1+1/K)n$.
Our choices $\eta=1/(Kn)$ and $\tau=\eta/(KHn^2)$ give
$\Delta=1/(K^6H^2n^7)$.
With $K,H$ fixed, every finite prefix has $O(n^7)$ updates,
so the algorithm must terminate.
\end{proof}

It remains to turn the tracked row bounds into a bound for the
final signing. The discrepancy has three contributions: the
tracked row sum, the error from dominant-entry removal, and the
rounding of the last $K$ coordinates.

\begin{proof}[Proof of Theorem~\ref{thm:main}]
If $n\le K$, all signs $+1$ give discrepancy at most $K$, since
$|a_{ij}|\le1$. For $n>K$, Proposition~\ref{prop:progress}
gives $s\le K$ at termination. Then $q_i\le s<K^2$, so no
row is large. For medium rows, the two nonpositive barriers in
\eqref{eq:barrier} give $|d_i|\le2K^3$.
Completed rows retain this bound from their last large or medium
state; rows completed at initialization have $d_i=0$.

Lemma~\ref{lem:deletion} bounds the tracking error by $4K$.
Rounding at most $K$ coordinates to the nearest sign changes
each by at most $1$, hence any row sum by at most $K$.
Thus $\|A\varepsilon\|_\infty\le2K^3+4K+K=8272$.
Proposition~\ref{prop:arithmetic} gives the deterministic
arithmetic bound.
\end{proof}

\begin{proof}[Proof of Corollary~\ref{cor:beck-fiala}]
Set $t_0=\min\{t,m\}$. Starting from $q=1$, double $q$ until
$q^2\ge t_0$; this takes $O(\log(2+m))$ arithmetic operations
and comparisons and gives
$\sqrt{t_0}\le q<2\sqrt{t_0}\le2\sqrt t$.
The columns of $B/q$ have Euclidean norm at most $1$, so
Theorem~\ref{thm:main} returns a signing with
$\|B\varepsilon\|_\infty\le8272q<16544\sqrt t$.
\end{proof}

\paragraph{AI methodology.}
The Odin Automatic AI Research Agent was used to find
the initial proof and assisted the authors in refining
the arguments and revising the manuscript.
The authors verified the final proofs and take full
responsibility for the content of this paper.

\paragraph{Acknowledgment.}
The authors are grateful to Professor Joel Spencer
for his encouragement and helpful discussions.

\clearpage
\appendix

\section{Proof of the Cleanup Lemma}
\label{app:cleanup}

\begin{proof}[Proof of Lemma~\ref{lem:cleanup}]
We must show that cleanup keeps the scales valid, preserves the
barrier bounds, and cannot increase $L$. The main step is to compare
the Gram-matrix contributions when a row becomes medium or remains
medium. We first verify the scale invariant \eqref{eq:scale}. The candidate
$r_i^{\rm cand}=2^{k_i}\max_{j\in T_i}|a_{ij}|$ lies in $[\sqrt{q_i},2\sqrt{q_i})$.
For a medium row, $K\ge\sqrt{q_i}$, and any previous scale is at least
$\sqrt{q_i}$ because $q_i$ only decreases. Taking the minimum
therefore preserves \eqref{eq:scale} and never increases the scale.

Consider a row that has just become medium. We compare its new
Gram-matrix terms with the diagonal contribution it had while large. Weighted Cauchy--Schwarz gives
$pp^{\mathsf T}\preceq(\sum_jp_j)\Diag(p)$ for $p\ge0$.
This follows by testing on a vector $y$:
$(p^{\mathsf T}y)^2\le(\sum_jp_j)(\sum_jp_jy_j^2)$.
Since $q_i\le r_i^2$, the two new terms in $W$ will be bounded
by $\Diag(p_i)/K$ if
$\gamma_{i\sigma}:=2r_i^2w_{i\sigma}\le1/K$ for both signs.
A new medium row has $q_i\le K^2$ and $d_i=0$: the latter
holds at initialization, and thereafter is inherited from the large
class because removal preserves $d_i$.
Setting $t=K/r_i\ge1$, we obtain
\[
 \log\gamma_{i\sigma}
 \le2\log t-Kt+6+\frac1{4K}
 \le-\frac K2<-\log K.
\]
The middle expression has derivative $2/t-K<0$ for $t\ge1$;
the last inequalities
hold for $K\ge16$. Hence $\gamma_{i\sigma}<1/K$, which bounds
the sum of the two new terms by $\Diag(p_i)/K$.
For a formerly large row, this is bounded by its previous contribution
to $\mathcal B$, since the retained entries only decrease.
The same row bounds give $u_{i\sigma}\le-2K^2+12K+1/2<-1$,
so each new barrier starts below the nearly tight threshold.
Initially $\sum_i p_{ij}\le1$, so
$W+\mathcal B\preceq I/K$ and $L\le1/K+\eta n=2/K$.
At $x=0$, this also gives $\Phi=0$.

For an existing medium row, we show that its barriers and
Gram-matrix contribution cannot increase. For scale changes, we use the monotonicity
idea in~\cite[Lemma~3.8(c) and Remark~3.9,
arXiv v2]{GK26}, adapted to our retained-entry
weights. We verify it directly below.
Removal preserves $d_i$ and cannot increase $G_i$.
At fixed scale, neither $u_{i\sigma}$ nor $w_{i\sigma}$ increases.
Keeping the new $G_i$ and $d_i$ fixed, and suppressing the row
and sign indices, differentiation gives
\[
 \frac{\partial u}{\partial r}
 =\frac{12K-u-G/(2r^2)}r,\qquad
 \frac{\partial\log w}{\partial r}
 =\frac{2-(u+G/(2r^2))/(2K)}r.
\]
Both expressions are positive when $u\le0$ and $G\le r^2$.
For every $r\in[r^{\rm new},r^{\rm old}]$,
$G^{\rm new}\le q^{\rm new}\le(r^{\rm new})^2\le r^2$,
so the latter inequality holds throughout the scale change.
As $r$ decreases, $u$ decreases whenever $u\le0$, including
at $u=0$. Thus $u$ cannot cross from this region into $u>0$,
and neither $u$ nor $w$ increases during the scale change.
The vectors $p_i$ also do not increase entrywise. Hence each existing
medium contribution is entrywise nonincreasing.

We now combine the rowwise comparisons to bound $L$. We use the
monotonicity of the top eigenvalue for symmetric matrices with
nonnegative entries:
replacing a Rayleigh vector by its coordinatewise
absolute value never decreases the quadratic form, and entrywise
decrease cannot increase the form on a nonnegative vector.
Apply the comparisons in stages. Freezing first takes a principal
submatrix. Updating rows that remain medium is then an entrywise
decrease between symmetric matrices with nonnegative entries.
Next update the large-row terms: diagonal reductions and replacements
by medium-row terms decrease the matrix in positive semidefinite order,
by the first part of the proof. Removing completed-row contributions
also decreases it in that order. Thus no stage increases $L$.
The quadratic term is unchanged at fixed $x$, so $\Phi$ cannot increase.
\end{proof}

\section{Proof of the Restricted-Trace Lemma}\label{sec:linear-algebra}

The key is to normalize by the diagonal matrix $D$. This turns
the coordinatewise hypothesis into a trace bound, which can then
be compared with the dimension of the given subspace.

\begin{proof}[Proof of Lemma~\ref{lem:restricted-trace}]
We follow \cite[Lemma~4.2]{GK26}, using the diagonal
normalization from \cite[Appendix~B.2]{LRR17}.
First handle the coordinates where $D$ vanishes.
If $D_{jj}=0$, then $S_{jj}=0$, and positivity gives
$|S_{jk}|^2\le S_{jj}S_{kk}=0$.
Thus $\ker D\subseteq\ker S$, which proves the lemma if
$\mathcal H\cap\ker D\ne\{0\}$.

Otherwise $D^{1/2}$ is injective on $\mathcal H$, so its image has the
same dimension $d=\dim \mathcal H$.
Restricting to the coordinates for which $D_{jj}>0$, the matrix
$T=D^{-1/2}SD^{-1/2}$ is positive semidefinite and
$\tr T=\sum_{j:D_{jj}>0}S_{jj}/D_{jj}\le\rho s$.
For an orthonormal basis $f_1,\ldots,f_d$ of $D^{1/2}\mathcal H$,
positivity implies
$\sum_{\ell=1}^d f_\ell^{\mathsf T}Tf_\ell\le\tr T<d$.
Choose a unit $f$ in this image with $f^{\mathsf T}Tf<1$,
and write $f=D^{1/2}h$ for $h\in \mathcal H$.
Since $\ker D\subseteq\ker S$, we have
$h^{\mathsf T}Sh=f^{\mathsf T}Tf<1=h^{\mathsf T}Dh$.
\end{proof}

\section{A Taylor Bound for the Potential}\label{sec:taylor}

Section~\ref{sec:progress} uses a third-order Taylor bound for
$\Phi$ along each update. During such a step, the active set,
retained entries, offsets, and row scales remain fixed.
The main issue is that the top eigenvector of $M$ also moves.
The spectral gap from Lemma~\ref{lem:spectral-gap} controls its
derivatives, giving a uniform remainder bound.

\begin{lemma}\label{lem:smoothness}
There is an absolute constant $H\ge1$ such that, at every state
after coefficient removal and scale updates, with $s>K$ and
$L\le2$, keeping
$V,T_i,c_i,r_i$ fixed,
\[
 \left|\Phi(x+th)-\Phi(x)-t\langle\nabla\Phi,h\rangle
       -\frac{t^2}{2}h^{\mathsf T}\nabla^2\Phi h\right|
 \le\frac{Hn^2}{6}|t|^3
\]
for every $h\in\R^s$ with $\|h\|_2\le1$ and every
$t\in[-1/4,1/4]$ for which the segment stays in the cube.
\end{lemma}
\begin{proof}
We bound the first three derivatives of $M$, then the derivatives
of its top eigenvalue $L$, and finally those of $\Phi$.
Fix a unit direction $h$. Primes denote derivatives along
$x+th$, and the indices on $u,w$ are suppressed. Throughout the
cube, \eqref{eq:gradient-bounds} gives
$|u'|\le2$, $|u''|\le1/K$, and $u'''=0$.
Writing $\beta=1/(2K)$, the first three weight derivatives
divided by $w$ are
$\beta u'$,
$\beta^2(u')^2+\beta u''$, and
$\beta^3(u')^3+3\beta^2u'u''$.
Each has absolute value at most one.
At a point $y=x+th$ on the segment, the change in $u$ is at most
$1/2$, since $|t|\le1/4$. Each weight therefore increases by
a factor less than two, so $M(y)\preceq2M(x)$ and $L(y)\le4$.
For the first three derivatives along this line, positivity of each
matrix summand gives
$-W\preceq M^{(j)}\preceq W$ for $j=1,2,3$.
Consequently $\|M^{(j)}\|_{\op}\le4$.

The eigenvector can vary even when the matrix derivatives are
small. To control this variation, we use the spectral gap.
Let $v$ be the positive unit top eigenvector. The gap proof in
Lemma~\ref{lem:spectral-gap} applies throughout the segment, so $v$ can
be chosen smoothly along it. The matrix $LI-M$ is invertible
on $v^\perp$, with inverse norm at most $1/\eta=Kn$.
Let $R_*$ denote that inverse extended by zero on the span of $v$;
this is the reduced resolvent.
Writing $v''_\perp$ for the projection of $v''$ onto $v^\perp$,
differentiation of $Mv=Lv$ and $v^{\mathsf T}v=1$ gives
\[
 v'=R_*M'v,\qquad
 v''_\perp=R_*\bigl(M''v+2(M'-L'I)v'\bigr),\qquad
 v^{\mathsf T}v''=-\|v'\|_2^2.
\]
The identity $L'=v^{\mathsf T}M'v$ gives $|L'|\le4$.
Together with $\|R_*\|_{\op}\le Kn$, the preceding identities give
$\|v'\|_2\le4Kn$ and
$\|v''\|_2\le4Kn+80K^2n^2$.
Successive differentiation of $L'=v^{\mathsf T}M'v$ yields
\[
 \begin{aligned}
 L''&=v^{\mathsf T}M''v+2(v')^{\mathsf T}M'v,\\
 L'''&=v^{\mathsf T}M'''v+4(v')^{\mathsf T}M''v
          +2(v'')^{\mathsf T}M'v+2(v')^{\mathsf T}M'v'.
 \end{aligned}
\]
Substituting the matrix and eigenvector bounds yields
$|L'|\le4$, $|L''|\le4+32Kn$, and
$|L'''|\le4+96Kn+768K^2n^2$.
For $f(t)=(t-1)_+^3$, the derivatives on $[0,4]$ obey
$|f'|\le27$, $|f''|\le18$, and $|f'''|\le6$ where defined.
Thus
\[
 |(f\circ L)'''|
 \le6|L'|^3+54|L'||L''|+27|L'''|
 =O(n^2).
\]
The negative quadratic in $\Phi$ has zero third derivative, so
this also bounds the third derivative of the potential wherever it exists.
Along the segment, $L$ is smooth with the derivative bounds above.
Since $f''$ is Lipschitz, $(f\circ L)''$ is absolutely
continuous, including across $L=1$. The displayed almost-everywhere
bound makes $(f\circ L)''$ Lipschitz with constant $Hn^2$.
For the fixed choice $K=16$, the explicit estimates above give
$|(f\circ L)'''|<2^{24}n^2$, so $H=2^{24}$ is admissible in
$\tau=\frac{\eta}{KHn^2}, \epsilon=\frac{\eta\tau}{K^2}$. Integrating the difference of second derivatives
twice gives the stated Taylor remainder.

For $\|h\|_2\le1$, the case $h=0$ is immediate.
Otherwise apply the unit-direction estimate to $h/\|h\|_2$
with step length $t\|h\|_2$. The remainder gains a factor
$\|h\|_2^3\le1$.
\end{proof}

\section{Proof of the Direction-Finding Lemma}\label{sec:computation}

Proposition~\ref{prop:direction} guarantees a
negative-curvature direction. To find one deterministically,
we approximate the gradient and Hessian, then use elimination
to obtain a direction satisfying the linear constraints.
Throughout, $V,T_i,c_i,r_i$ remain fixed when differentiating; only the active coordinates
$x_V$ vary. The constant $K$ is fixed, and all constants
in the estimates are absolute.

\subsection{Approximating the gradient and Hessian}

We first establish the accuracy needed by the direction-finding
algorithm. We seek approximations $\widetilde g$ and
$\widetilde H$ satisfying
\begin{equation}\label{eq:approximation}
\|\widetilde g-\nabla\Phi\|_2\le\epsilon,
\qquad
\|\widetilde H-\nabla^2\Phi\|_{\op}
\le\epsilon.
\end{equation}

\begin{lemma}\label{lem:numerical}
At a state after coefficient removal and scale updates, with
$s>K$, $L\le2$ and
$u_{i\sigma}\le0$ for every medium pair, one can compute
$\widetilde g$ and a symmetric $\widetilde H$ satisfying
\eqref{eq:approximation} in
$O(mn^2+(n^3+m)\log(2mn/\epsilon))$
real arithmetic operations and comparisons, for $0<\epsilon<1$.
\end{lemma}
\begin{proof}
We first identify the quantities needed for the derivatives.
Let $(L,v)$ be the top eigenpair of $M$, with $\|v\|_2=1$,
and write $M_j=\partial_jM$ and
$M_{jk}=\partial_j\partial_kM$.
Let $R_*$ be the inverse of $LI-M$ on $v^\perp$,
extended by zero on $v$.
By Lemma~\ref{lem:spectral-gap}, $L$ is simple, and
standard eigenvalue differentiation gives
\begin{equation}\label{eq:spectral-derivatives}
 L_j=v^{\mathsf T}M_jv,\qquad
 L_{jk}=v^{\mathsf T}M_{jk}v+
        2(M_jv)^{\mathsf T}R_*(M_kv).
\end{equation}
For $f(t)=(t-1)_+^3$, the chain rule gives
\[
\begin{aligned}
 \nabla\Phi&=f'(L)\nabla L-2\eta x_V,\\
 \nabla^2\Phi&=f'(L)\nabla^2L+
 f''(L)\nabla L\nabla L^{\mathsf T}-2\eta I.
\end{aligned}
\]
Thus we need approximations to the weights defining $M$,
its top eigenpair, and the reduced resolvent $R_*$.
We build them using arithmetic and comparisons, then
substitute them into these formulas.

Set $\kappa=2^{-40}\epsilon/n^6$ and
$\rho=\kappa/[K^2(m+1)]$.
Our goal is to approximate $M,M_j,M_{jk}$ within
$\kappa$, starting from exponential errors at most $\rho$.
The spectral gap controls how these matrix errors affect the
eigenpair and $R_*$; the remaining margin in $\kappa$
ensures the final errors are at most $\epsilon$.
The matrices $M_j,M_{jk}$ are used only in the analysis;
the implementation computes their required products directly.

\smallskip
\noindent\textbf{Exponential weights.}
Write $z_i=p_i/r_i^2$, so $\|z_i\|_1\le1$ and
\[
 W=\frac{K^2}{2}\sum_{i,\sigma}
 e^{u_{i\sigma}/(2K)}z_iz_i^{\mathsf T}.
\]
Replace each exponential by a nonnegative approximation
with absolute error at most $\rho$. Using the same values
in $M,M_j,M_{jk}$ gives hatted matrices.
For $\beta=1/(2K)$, the derivative multipliers are
$\beta(g_{i\sigma})_j$ and
$\beta^2(g_{i\sigma})_j(g_{i\sigma})_k+
\beta(\nabla^2u_{i\sigma})_{jk}$,
each of magnitude at most one by
\eqref{eq:gradient-bounds}.
There are at most $2m$ rank-one terms,
each of norm at most one before its scalar weight.
Their total error is at most $K^2m\rho\le\kappa$, giving
\begin{equation}\label{eq:matrix-errors}
 \|\widehat M-M\|_{\op},\quad
 \|\widehat M_j-M_j\|_{\op},\quad
 \|\widehat M_{jk}-M_{jk}\|_{\op}\le\kappa.
\end{equation}
These matrices approximate the exact derivatives of $M$;
they need not be derivatives of the numerical routine.
The exact diagonal and rank-one terms, together with
nonnegative weights, keep $\widehat M$ positive
semidefinite and entrywise positive, with gap at least
$\eta$ and top eigenvalue at most $L+\kappa<3$.

Here is a finite-arithmetic approximation to the exponentials.
All their arguments are nonpositive.
Start with $J=1$ and $q_{\rm exp}=1/2$;
repeatedly halve $q_{\rm exp}$ and increment $J$ until
$q_{\rm exp}\le\rho/4$.
For an argument below $-J$, return zero; its error is
at most $e^{-J}<2^{-J}\le\rho/4$.
On $[-J,0]$, evaluate the Taylor polynomial at zero
of degree $32(J+1)$, clipping a negative result to zero.
With $N=32(J+1)+1$, the remainder is at most
\[
 e^J J^N/N!\le e^J(eJ/N)^N
 <2^{2J-3N}\le2^{-J}\le\rho/4.
\]
Clipping cannot increase the absolute error.
Horner evaluation costs $O(J)=O(\log(1/\rho))$
operations, independently of the argument's magnitude.

\smallskip
\noindent\textbf{Leading eigenpair.}
Next approximate the top eigenpair.
For the error analysis, let $(\widehat L,\widehat v)$
be the exact top eigenpair of $\widehat M$, with a
positive unit eigenvector. The routine computes
$(\ell,\widetilde v)$ by power iteration and normalization.
Start the iteration from $\ones$, dividing each iterate
by its largest coordinate. This keeps its direction
unchanged and its squared norm in $[1,s]$.
The initial unit vector has overlap at least
$1/\sqrt s$ with $\widehat v$.
By semidefiniteness and the gap, every non-top
eigenvalue has ratio at most $1-1/(3Kn)$
to $\widehat L$.
After $T$ iterations, the angle with $\widehat v$
has tangent at most $\sqrt s\,e^{-T/(3Kn)}$.

Starting at $q_{\rm pow}=1,J_{\rm pow}=0$, double
$q_{\rm pow}$ and increment $J_{\rm pow}$ until
$q_{\rm pow}\ge16n/\kappa$.
With $T=12KnJ_{\rm pow}$, the tangent is at most
$\kappa/16$.
Bisect the inverse norm of the final iterate $y$
on $[0,1]$ to accuracy $\kappa/(8n)$,
comparing $a^2\|y\|_2^2$ with one at each trial.
The resulting vector $\widetilde v$ satisfies
$\|\widetilde v-\widehat v\|_2\le\kappa$,
although it need not be exactly unit length.
Its Rayleigh quotient
$\ell=\widetilde v^{\mathsf T}\widehat M\widetilde v/
(\widetilde v^{\mathsf T}\widetilde v)$
differs from $\widehat L$ by at most
$3(\kappa/16)^2\le\kappa$.

To transfer this estimate to $M$, use
$|\widehat L-L|\le\kappa$ and $\kappa\le\eta/2$.
Projecting
$(\widehat L I-M)\widehat v
 =(\widehat M-M)\widehat v$
onto $v^\perp$ bounds the projected norm of
$\widehat v$ by $\kappa/(\eta-\kappa)$.
The positive orientations of the eigenvectors
bound their distance by twice this norm.
Consequently,
\begin{equation}\label{eq:eigenpair-errors}
 |\ell-L|\le2\kappa,\qquad
 \|\widetilde v-v\|_2
 \le(4Kn+1)\kappa\le2^7n\kappa.
\end{equation}

\smallskip
\noindent\textbf{Resolvent and derivative errors.}
We now approximate the remaining quantity $R_*$.
Set $Z=LI-M+vv^{\mathsf T}$.
It acts as the identity on $v$ and as $LI-M$ on
$v^\perp$, so $\|Z^{-1}\|_{\op}\le16n$ and
$R_*=Z^{-1}-vv^{\mathsf T}$.
For
$\widehat Z=\ell I-\widehat M+
\widetilde v\widetilde v^{\mathsf T}$,
we have $\|\widehat Z-Z\|_{\op}
\le3\kappa+3\|\widetilde v-v\|_2
\le2^9n\kappa$.
Hence $\|Z^{-1}(\widehat Z-Z)\|_{\op}
\le2^{13}n^2\kappa<1/2$,
and Gaussian elimination computes $\widehat Z^{-1}$,
with norm at most $32n$.
Set $\widehat R_*=\widehat Z^{-1}
-\widetilde v\widetilde v^{\mathsf T}$.
The inverse-difference identity gives
\[
 \|\widehat R_*-R_*\|_{\op}
 \le(32n)(2^9n\kappa)(16n)
     +3\|\widetilde v-v\|_2
 \le2^{19}n^3\kappa,
\]
and $\|\widehat R_*\|_{\op}\le32n+4\le2^6n$.

We can now propagate these estimates through
\eqref{eq:spectral-derivatives}.
Let $y_j=M_jv$ and
$\widehat y_j=\widehat M_j\widetilde v$.
The base-state derivative bounds imply
$\|M_j\|_{\op},\|M_{jk}\|_{\op}\le2$,
$\|y_j\|_2\le2$, $\|\widehat y_j\|_2\le6$, and
$\|\widehat y_j-y_j\|_2
\le2\kappa+2\|\widetilde v-v\|_2
\le2^9n\kappa$.
To compare the resolvent terms, insert and subtract
$y_j^{\mathsf T}\widehat R_*\widehat y_k$ and
$y_j^{\mathsf T}\widehat R_*y_k$.
The two vector-approximation terms are bounded by
$2^{18}n^2\kappa$ each; the remaining resolvent-error
term is at most $2^{21}n^3\kappa$.
The direct Rayleigh terms contribute at most
$6\|\widetilde v-v\|_2+4\kappa\le2^{10}n\kappa$.
Multiplying coordinate errors by $\sqrt s$
and entry errors by $s$ gives
\[
 \|\widehat{\nabla L}-\nabla L\|_2
 \le2^{10}n^{3/2}\kappa,\qquad
 \|\widehat{\nabla^2L}-\nabla^2L\|_{\op}
 \le2^{24}n^4\kappa.
\]
Here the hatted derivatives are obtained by substituting
$\widehat M_j,\widehat M_{jk},\ell,\widetilde v,
\widehat R_*$ into \eqref{eq:spectral-derivatives}.

Finally, use the chain-rule formulas at the start of the proof.
The base-state estimates give $\|\nabla L\|_2\le2$,
$\|\nabla^2L\|_{\op}\le2^8n$, and
$\|\widehat{\nabla L}\|_2\le3$.
Both $L,\ell\in[0,3]$, where $f',f''$ and the
Lipschitz constant of $f'$ are at most $2^4$,
and that of $f''$ is at most $2^3$.
The resulting gradient and Hessian errors are at most
$2^{15}n^2\kappa$ and $2^{29}n^4\kappa$.
By the choice of $\kappa$, these are at most
$2^{-25}\epsilon/n^4$ and $2^{-11}\epsilon/n^2$,
so \eqref{eq:approximation} follows.

\Needspace{11\baselineskip}
\smallskip
\noindent\textbf{Computing the required products.}
It remains to implement the substituted derivatives without
forming every $M_j$ and $M_{jk}$.
Let $\widehat\alpha_{i\sigma}$ be $K^2/2$ times
the chosen exponential approximation, and put
$\nu_i=z_i^{\mathsf T}\widetilde v$.
Define $\widehat g_L$, the matrix $Y$ with columns
$\widehat M_j\widetilde v$, and $S_L$ with entries
$(S_L)_{jk}=\widetilde v^{\mathsf T}
\widehat M_{jk}\widetilde v$ by the outer-product sums
\[
 \begin{aligned}
 \widehat g_L
   &=\sum_{i,\sigma}\widehat\alpha_{i\sigma}
     \beta\nu_i^2g_{i\sigma},\\
 Y&=\sum_{i,\sigma}\widehat\alpha_{i\sigma}
     \beta\nu_i z_i g_{i\sigma}^{\mathsf T},\\
 S_L&=\sum_{i,\sigma}\widehat\alpha_{i\sigma}\nu_i^2
    \left(\beta^2g_{i\sigma}g_{i\sigma}^{\mathsf T}
          -\frac{\beta}{r_i^2}\Diag(p_i)\right).
 \end{aligned}
\]
Then $\widehat{\nabla L}=\widehat g_L$ and
$\widehat{\nabla^2L}=S_L+2Y^{\mathsf T}\widehat R_*Y$.
The resulting Hessian approximation is symmetric.

The outer-product sums cost $O(mn^2)$ operations;
dense products and Gaussian elimination cost $O(n^3)$.
Exponential evaluation, power iteration, and norm bisection
cost $O(m\log(m/\kappa))$,
$O(n^3\log(n/\kappa))$, and
$O(\log(n/\kappa))$, respectively.
Since $\kappa=2^{-40}\epsilon/n^6$,
the total is
$O(mn^2+(n^3+m)\log(2mn/\epsilon))$
real arithmetic operations and comparisons.
\end{proof}

\subsection{Finding a negative quadratic direction}

The approximate Hessian will be restricted to a space of
directions satisfying the linear constraints. We use the next
elementary procedure to extract a negative direction without
an eigenvalue algorithm.

\begin{lemma}\label{lem:elimination}
Given a symmetric matrix $Q\in\R^{s\times s}$ with a negative
eigenvalue, a vector $z$ with $z^{\mathsf T}Qz<0$ can be found
in $O(s^3)$ arithmetic operations and comparisons.
\end{lemma}
\begin{proof}
We eliminate positive diagonal pivots by Schur complements,
recovering a negative vector by back-substitution.
A negative diagonal entry immediately gives a coordinate vector.
Otherwise, if a positive diagonal entry is present, permute it to
the first position and write
$Q=\left(\begin{smallmatrix}a&b^{\mathsf T}\\b&C\end{smallmatrix}\right)$,
with $a>0$. The identity
\[
 \begin{pmatrix}-b^{\mathsf T}y/a\\y\end{pmatrix}^{\mathsf T}
 Q
 \begin{pmatrix}-b^{\mathsf T}y/a\\y\end{pmatrix}
 =y^{\mathsf T}(C-bb^{\mathsf T}/a)y
\]
reduces the search to the smaller matrix $C-bb^{\mathsf T}/a$,
the Schur complement. Completing the square separates a positive
square with coefficient $a$ from this smaller quadratic form,
which must therefore still have a negative eigenvalue. Repeat,
then recover the original coordinates using the displayed substitution
at each elimination step.
If every remaining diagonal entry is zero, choose the least pair
$(j,k)$, with $j<k$, for which the residual entry $Q_{jk}\ne0$.
Use $e_j-e_k$ when $Q_{jk}>0$, and $e_j+e_k$ when
$Q_{jk}<0$; the quadratic form is $-2|Q_{jk}|$.
A zero residual matrix would make the original matrix positive
semidefinite, contrary to the hypothesis.
Choose the least eligible pivot or index pair at each step. The
Schur-complement updates take $O(s^3)$ operations in total.
\end{proof}

\subsection{Completing the direction search}
\label{sec:direction-implementation}

\begin{proof}[Proof of Lemma~\ref{lem:direction-finding}]
We restrict the approximate Hessian to the directions
satisfying the row constraints and the approximate first-order
condition. The existence result from Section~\ref{sec:curvature}
ensures a negative quadratic direction in this space.
First compute approximations $\widetilde g$ and
$\widetilde H$ satisfying \eqref{eq:approximation},
using Lemma~\ref{lem:numerical}.
Form a matrix $E$ whose rows are the retained
row vectors of all large rows, the gradients
of all nearly tight barriers, and
$\widetilde g^{\mathsf T}$.
Let $P$ be the orthogonal projector onto $\ker E$.
It can be computed by retaining a basis of the
rows of $E$ and applying Gaussian elimination.

By Proposition~\ref{prop:direction}, applied with
$g_*=\widetilde g$, there is a unit vector
$u\in\ker E$ such that
$u^{\mathsf T}\nabla^2\Phi\,u\le-2\eta$.
Since $\|\widetilde H-\nabla^2\Phi\|_{\op}
\le\epsilon<\eta/4$, it follows that
\[
 u^{\mathsf T}(\widetilde H+\eta I)u
 \le-\eta+\epsilon<0.
\]
Thus the symmetric matrix
\[
 Q=P(\widetilde H+\eta I)P+(I-P)
\]
has a negative eigenvalue.
Lemma~\ref{lem:elimination} finds a vector $z$
with $z^{\mathsf T}Qz<0$ by deterministic
real arithmetic and comparisons.
For $h_0=Pz$, we have
\[
 h_0^{\mathsf T}(\widetilde H+\eta I)h_0
 \le z^{\mathsf T}Qz<0,
\]
so $h_0\ne0$ and $h_0\in\ker E$.

Divide $h_0$ by $\|h_0\|_\infty$ and then
repeatedly halve the vector until its Euclidean
norm is at most $1$.  The resulting vector $h$
satisfies $1/2\le\|h\|_2\le1$ and lies in
$\ker E$.  In particular, all row constraints
hold and $\widetilde g^{\mathsf T}h=0$.
The approximation guarantees give
\[
 |\langle\nabla\Phi,h\rangle|
 \le\epsilon\|h\|_2\le\epsilon,
 \qquad
 h^{\mathsf T}\nabla^2\Phi\,h
 <-(\eta-\epsilon)\|h\|_2^2
 \le-\frac{3\eta}{16}.
\]
The derivative approximations take
$O(mn^2+(n^3+m)\log(2mn/\epsilon))$
operations by Lemma~\ref{lem:numerical}.
Computing the constraint basis, projector,
and negative direction takes $O(mn^2+n^3)$
operations; normalization takes $O(\log(2n))$
halvings.  All steps are deterministic, and
$\epsilon$ is inverse-polynomial in $n$.
Hence the total cost is polynomial in $m,n$.
\end{proof}

\section{Arithmetic Complexity}\label{sec:complexity}

We combine the cost of one direction search with the iteration
bound of Proposition~\ref{prop:progress}.
The routines use the operations specified in Section~\ref{sec:algorithm}.
Scale updates and direction normalization take $O(\log(2n))$
doublings or halvings. For scales, $q_i\le s\max_jp_{ij}$ bounds
the ratio to the starting squared scale; for directions, the squared
norm starts in $[1,s]$.
Lemma~\ref{lem:numerical} approximates the derivatives without forming
every derivative matrix separately. The per-iteration costs are
\[
\begin{array}{ll}
 \text{Derivative approximation}
   &O\bigl(mn^2+(n^3+m)\log(2mn/\epsilon)\bigr),\\
 \text{Constraint basis, projection and elimination}
   &O(mn^2+n^3),\\
 \text{Row quantities, scales and initial removal scan}
   &O(mn).
\end{array}
\]
For coefficient removal, process rows in increasing index order.
Each deliberate deletion requires at most one further $O(n)$ scan
of its row. Every entry is removed at most once, so these extra scans
cost $O(mn^2)$ over the entire run.

By $\tau=\frac{\eta}{KHn^2}, \epsilon=\frac{\eta\tau}{K^2}$, $\epsilon$ is an absolute constant times
$n^{-4}$. Proposition~\ref{prop:progress} bounds the number of
updates by $N=O(n^7)$. The total cost is therefore
\begin{equation}\label{eq:total-cost}
 \begin{aligned}
 &N\,O\bigl(mn^2+(n^3+m)\log(2mn/\epsilon)\bigr)+O(mn^2)\\
 &\hspace{2em}=O\bigl((mn^9+n^{10})\log(2+m+n)\bigr).
 \end{aligned}
\end{equation} 

\clearpage

\bibliographystyle{alpha-labelprefix}
\bibliography{komlos}
\end{document}